\def\ARXIV{}

\documentclass[sigconf]{aamas} 

\usepackage{balance} 
\usepackage[linesnumbered,lined,boxed,ruled,vlined]{algorithm2e} 

\SetKwComment{Comment}{$\triangleright$\ }{}

\theoremstyle{definition}
\newtheorem{remark}{Remark}[section]
\newtheorem{corollary}{Corollary}[section]
\newtheorem{theorem}{Theorem}[section]
\newtheorem{lemma}{Lemma}[section]

\newtheorem{proposition}{Proposition}[section]

\newtheorem{example}{Example}[section]

\newcommand{\ins}[1]{\langle #1 \rangle}
\newcommand{\ceil}[1]{\left\lceil #1 \right\rceil}
\newcommand{\floor}[1]{\left\lfloor #1 \right\rfloor}

\definecolor{ForestGreen}{rgb}{.13,.54,.13}

\usepackage[capitalise,noabbrev]{cleveref}
\Crefname{remark}{Remark}{Remarks}
\Crefname{rmk}{Remark}{Remarks}
\Crefname{dfn}{Definition}{Definitions}
\Crefname{thm}{Theorem}{Theorems}
\Crefname{cor}{Corollary}{Corollaries}
\Crefname{lem}{Lemma}{Lemmas}
\Crefname{examplex}{Example}{Examples}
\Crefname{prop}{Proposition}{Propositions}

\usepackage{tikz}
\newcommand*\circled[1]{\tikz[baseline=(char.base)]{
            \node[shape=circle,draw,inner sep=1pt, minimum size=1.5em] (char) {#1};}}

\newcommand{\MMS}[2]{\text{MMS}^{#1}_{#2}}

\ifdefined\ARXIV
\acmConference[AAMAS '22]{AAMAS '22}
\copyrightyear{2022}
\acmDOI{}
\acmPrice{}
\acmISBN{}
\else
\setcopyright{ifaamas}
\acmConference[AAMAS '22]{Proc.\@ of the 21st International Conference
on Autonomous Agents and Multiagent Systems (AAMAS 2022)}{May 9--13, 2022}
{Auckland, New Zealand}{P.~Faliszewski, V.~Mascardi, C.~Pelachaud,
M.E.~Taylor (eds.)}
\copyrightyear{2022}
\acmYear{2022}
\acmDOI{}
\acmPrice{}
\acmISBN{}
\fi

\acmSubmissionID{206}

\title[MMS for Chores]{Ordinal Maximin Share Approximation for Chores}

\author{Hadi Hosseini}
\affiliation{
  \institution{Pennsylvania State University}
   \country{}
  }
\email{hadi@psu.edu}

\author{Andrew Searns}
\affiliation{
  \institution{Johns Hopkins University Applied Physics Laboratory}
   \country{}
  }
\email{andrew.searns@jhuapl.edu}

\author{Erel Segal-Halevi}
\affiliation{
  \institution{Ariel University}
   \country{}
  }
\email{erelsgl@gmail.com}

\begin{abstract}
We study the problem of fairly allocating a set of $m$ indivisible chores (items with non-positive value) to $n$ agents. We consider the desirable fairness notion of $1$-out-of-$d$ maximin share (MMS)---the minimum value that an agent can guarantee by partitioning items into $d$ bundles and receiving the least valued bundle---and focus on ordinal approximation of MMS that aims at finding the largest $d\leq n $ for which $1$-out-of-$d$ MMS allocation exists.
Our main contribution is a polynomial-time algorithm for $1$-out-of-$\lfloor\frac{2n}{3}\rfloor$ MMS allocation, and a proof of existence of $1$-out-of-$\lfloor\frac{3n}{4}\rfloor$ MMS allocation of chores. 
Furthermore, we show how to use recently-developed algorithms for bin-packing to approximate the latter bound up to a logarithmic factor in polynomial time.
\end{abstract}


\newcommand{\BibTeX}{\rm B\kern-.05em{\sc i\kern-.025em b}\kern-.08em\TeX}

\begin{document}


\pagestyle{fancy}
\fancyhead{}


\maketitle 


\section{Introduction}
Fairness is one of the most fundamental requirements in many multiagent systems. Fair division, in particular, deals with allocation of resources and alternatives in a fair manner by cutting across a variety of fields including computer science, economics, and artificial intelligence.
Traditionally, fair division has been concerned with the allocation of \textit{goods} that are positively valued by agents, leading to a plethora of fairness notions, axiomatic results, and computational studies (see \cite{brandt2016handbook} and \cite{doi:10.1146/annurev-economics-080218-025559} for detailed discussions).
However, many practical problems require the distribution of a set of negatively valued items (aka \textit{chores}). These problems range from assigning household chores or distributing cumbersome tasks to those involving collective ownership responsibility \cite{risse2008should} in human-induced factors such as climate change \cite{traxler2002fair}, nuclear waste management, or controlling gas emissions \cite{caney2009justice}.
The problem of allocating chores is crucially different from allocating goods both from axiomatic and computational perspectives. For instance, while goods are freely disposable, chores must be completely allocated.
These fundamental differences have motivated a large number of recent works in fair division of divisible \cite{bogomolnaia2019dividing,chaudhury2020dividing} and indivisible chores \citep{aziz2019fair,freeman2020equitable,ijcai2019-7,Aziz_Rauchecker_Schryen_Walsh_2017}.

When dealing with indivisible items, a compelling fairness notion is the Maximin Share (MMS) guarantee---proposed by \citet{budish2011combinatorial}---which is a generalization of the \textit{cut-and-choose protocol} to indivisible items \cite{brams1996fair}. 
An agent's $1$-out-of-$d$ maximin share value is the value that it can guarantee by partitioning $m$ items into $d$ bundles and receiving the least valued bundle. Unfortunately, the $1$-out-of-$n$ MMS allocations may neither exist for goods~\cite{kurokawa2018fair,feige2021tight} nor for chores~\cite{Aziz_Rauchecker_Schryen_Walsh_2017}. These non-existence results, along with computational intractability of computing such allocations, have motivated \textit{multiplicative} approximations of MMS wherein each agent receives an $\alpha \leq 1$ fraction of its $1$-out-of-$n$ MMS value when dealing with goods \citep{ghodsi2018fair,garg2020improved,garg2018approximating}, or $\alpha \geq 1$ approximation of its $1$-out-of-$n$ MMS value when dealing with chores \citep{Aziz_Rauchecker_Schryen_Walsh_2017,barman2017approximation,huang2021algorithmic}.

In this paper, we initiate the study of \textit{ordinal} MMS approximations for allocating chores. The goal is finding an integer $d \leq n$ for which $1$-out-of-$d$ MMS exists and can be computed efficiently. 
Recently, ordinal approximations of MMS for allocating `goods' have received particular attention as natural guarantees that provide a simple conceptual framework for justifying approximate decisions to participating agents: partition the items in a counterfactual world where there are $d \geq n$ agents available \cite{babaioff2019fair,babaioff2021competitive,segal2020competitive,hosseini2021mms,ElkindSeSu21}. 
Since these approximations rely on ordinal rankings of bundles, they are generally robust against slight changes in agent's valuation profiles compared to their multiplicative counterparts
(see \cref{app:robust} for an example and a detailed discussion).
Focusing on ordinal approximations, we discuss key technical differences between allocating goods and chores, and highlight practical computational contrasts between ordinal and multiplicative approximations of MMS.

\subsection{Contributions}
We make the following theoretical and algorithmic contributions.

\paragraph{\textbf{An algorithm for $1$-out-of-$\floor{\frac{2n}{3}}$ MMS}}
We show that heuristic techniques for allocating goods do not carry over to chores instances  (\cref{sec:reductions}), and develop other techniques to upper-bound the number of large chores (\cref{lem:counting}).
Using these techniques, we develop a greedy algorithm that achieves $1$-out-of-$\floor{\frac{2n}{3}}$ MMS approximation for chores (\cref{thm:Chores_TwoThirds}).
The algorithm runs in strongly-polynomial time: the number of operations required is polynomial in the number of agents and chores.

\paragraph{\textbf{Existence of $1$-out-of-$\floor{\frac{3n}{4}}$ MMS}}
 We show the existence of $1$-out-of-$\lfloor\frac{3n}{4}\rfloor$ MMS allocation of chores (\cref{thm:Chores_ThreeFourths}).
The main technical challenge is dealing with large chores that requires exact computation of MMS values, rendering our algorithmic approach intractable.
While our technique gives the best known ordinal approximation of MMS, it only provides a tight bound for small instances (\cref{prop:upper-bound-7}) but not necessarily for larger instances (\cref{prop:upper-bound-n}).

    
\paragraph{\textbf{Efficient approximation algorithm}}
We develop a practical algorithm for approximating
the $1$-out-of-$\lfloor\frac{3n}{4}\rfloor$ MMS bound for chores. More specifically, our algorithm guarantees $1$-out-of-$d$ MMS for $d = \floor{\floor{\frac{3n}{4}}-O(\log{n})}$ (\cref{thm:chores-approx}) and runs in time polynomial in the binary representation of the input.


\subsection{Related Work}

\subsubsection*{\textbf{MMS for allocating goods}}
The notion of maximin-share originated in the economics literature. \citet{budish2011combinatorial} showed a mechanism that guarantees $1$-out-of-$(n+1)$ MMS to all agents by adding a small number of excess goods. Whether or not $1$-out-of-$(n+1)$ MMS can be guaranteed without adding excess goods remains an open problem to date.

In the standard fair division settings, in which adding goods is impossible, the first non-trivial ordinal approximation was $1$-out-of-$(2n-2)$ MMS \citep{aigner2022envy}.
\citet{hosseini2021mms} studied the connection between guaranteeing 1-out-of-$n$ MMS for $2/3$ of the agents and the ordinal approximations for \textit{all} agents. The implication of their results is the existence of $1$-out-of-($\floor{3n/2}$) MMS allocations and a polynomial-time algorithm for $n<6$. Recently, a new algorithmic method has been proposed that achieves this bound for any number of agents \citep{hosseini2021ordinal}.
The ordinal approximations have been extended to $\ell$-out-of-$d$ MMS to guarantee that each agent receives at least as much as its worst $\ell$ bundles, where the goods were partitioned into $d$ bundles \citep{segal2019maximin,babaioff2019fair}.
The maximin share and its ordinal approximations have also been applied to some variants of the \emph{cake-cutting} problem \citep{bogomolnaia2020guarantees,ElkindSeSu21,ElkindSeSu21b,ElkindSeSu21c}.

The multiplicative approximation of MMS originated in the computer science literature \citep{procaccia2014fair}. 
These algorithms guarantee that each agent receives at least an $\alpha$ fraction of its maximin share threshold \citep{kurokawa2018fair,amanatidis2017approximation,garg2018approximating,ghodsi2018fair}. 
For goods, the best known existence result is $\alpha  \geq 3/4+1/(12n)$, and the best known polynomial-time algorithm 
guarantees $\alpha \geq 3/4$ \citep{garg2020improved}. The MMS bound was improved for special cases with only three agents \citep{amanatidis2017approximation}, and the best known approximation is $\alpha \geq 8/9$ \citep{gourves2019maximin}.

There are also MMS approximation algorithms for settings with constraints, such as when the goods are allocated on a cycle and each agent must get a connected bundle \citep{truszczynski2020maximin}.
\citet{mcglaughlin2020improving} showed an algorithm for approximating the maximum Nash welfare (the product of agents' utilities), which attains a fraction $1/(2n)$ of the MMS.
Recently, \citet{nguyen2017approximate} gave a Polynomial Time Approximation Scheme (PTAS) for a notion defined as \textit{optimal-MMS}, that is, the largest value, $\alpha$, for which each agent receives the value of $\alpha\cdot \text{MMS}_{i}$. Since the number of possible partitions is finite, an optimal-MMS allocation always exists, and it is an MMS allocation if $\alpha \geq 1$. However, an optimal-MMS allocation may provide an arbitrarily bad ordinal MMS guarantee \cite{Searns_Hosseini_2020,hosseini2021mms}.

\subsubsection*{\textbf{MMS for allocating chores}}
\citet{Aziz_Rauchecker_Schryen_Walsh_2017} initiated the study of MMS fairness for allocating indivisible chores. They proved that---similar to allocating goods---a $1$-out-of-$n$ MMS allocation may not always exist, and computing the MMS value for a single agent remains NP-hard.

In the maximin share allocation of chores, the multiplicative approximation factor is larger than $1$ (each agent might get a larger set of chores than its MMS value). The multiplicative factors in the literature have been improved from 2 \citep{Aziz_Rauchecker_Schryen_Walsh_2017} to 4/3 \citep{barman2017approximation} to 11/9 \citep{huang2021algorithmic}.
The best known polynomial-time algorithm guarantees a 5/4 factor \citep{huang2021algorithmic}.
\citet{aigner2022envy} prove the existence of a $1$-out-of-$\floor{2n/3}$ MMS allocation for chores, but their algorithm requires an exact computation of the MMS values, so it does not run in polynomial time.
Note that multiplicative and ordinal approximations do not imply one another---each of them might be better in some instances as we illustrate in the next example.

\begin{example}
\label{exm:ordinal-vs-cardinal}
Consider an instance with $n=3$ agents and $m$ identical chores of value $-1$. Then:
\begin{itemize}
\item If there are $m=2$ chores, then the $1$-out-of-$\floor{2n/3}$ MMS is $-1$, which is better than $11/9$ of the $1$-out-of-$n$ MMS.
\item If there are $m=3$ chores, then the $1$-out-of-$\floor{2n/3}$ MMS is $-2$, which is worse than $11/9$ of the $1$-out-of-$n$ MMS.
\end{itemize}
\end{example}
In Appendix \ref{app:relations} we generalize this example to any number of agents. Additionally, we study the relationships between the ordinal maximin share and other common fairness notions such as approximate-proportionality or approximate-envy-freeness. The bottom line is that all these notions are independent: none of them implies a meaningful approximation of the other.

The notion of maximin share fairness has been extended to \emph{asymmetric agents}, i.e. agents with different entitlements over chores \citep{aziz2019maxmin,ijcai2019-7}.
Recently, a variation of MMS has also been studied in conjunction with \emph{strategyproofness} that only elicits ordinal preferences as opposed to cardinal valuations \citep{ijcai2019-9,aziz2020approximate}.
In parallel, there are works studying other fairness notions for chores, or for combinations of goods and chores. Examples are 
approximate proportionality \citep{aziz2020polynomial},
approximate envy-freeness \citep{aziz2019fair},
approximate equitability \citep{freeman2020equitable},
and leximin \citep{chen2020fairness}. In the context of mixed items, however, no multiplicative approximation of MMS is guaranteed to exist \citep{kulkarni2021approximating}. In \cref{app:mixed} we show that similarly no ordinal MMS approximation is guaranteed to exist for mixed items.


\section{Preliminaries}

\paragraph{\textbf{Problem instance.}}
An instance of a fair division problem is denoted by $I = \ins{N, M, V}$ where $N = \{1,\ldots, n\}$ is a set of agents, $M = \{c_1, \ldots, c_m\}$ is a set of $m$ indivisible chores, and $V = (v_1, \ldots, v_n)$ is a valuation profile of agents. Agent $i$'s preferences over chores is specified by a valuation function $v_i: 2^M \to \mathbb{R}$.
We assume that the valuation functions are \textit{additive}; that is, for any agent $i\in N$, for each subset $S\subseteq M$, $v_{i}(S) = \sum_{c\in S} v_{i}(\{c\})$ where $v_i(\emptyset) = 0$. We assume items are chores for all agents, i.e., for each $i\in N$, for every $c\in M$ we have $v_i(\{c\}) \leq 0$.
For a single chore $c\in M$, we write $v_i(c)$ instead of $v_i(\{c\})$.
Without loss of generality, we assume that $m \geq n$ since otherwise we can add dummy chores that are valued $0$ by all agents.

\paragraph{\textbf{Allocation.}} An allocation $A = (A_1, \ldots, A_n)$ is an $n$-partition of the set of chores, $M$, where a bundle of chores $A_{i}$, possibly empty, is allocated to each agent $i\in N$. 
An allocation must be \textit{complete}: $\cup_{i\in N} A_{i} = M$.

\paragraph{\textbf{Maximin share.}}
Let $d \leq n$ be an integer and $\Pi_{d}(M)$ denote the set of $d$-partitions of $M$.
For each agent $i\in N$, the \textbf{$1$-out-of-$d$ Maximin Share} of $i$ on $M$, denoted $\text{MMS}_{i}^{d}(M)$, is defined as
$$
\text{MMS}_{i}^{d}(M) = \max_{(A_{1}, A_{2}, \ldots A_{d})\in \Pi_{d}(M)}\min_{j \in [d]} v_{i}(A_{j}),
$$
where $[d] = \{1,\ldots, d\}$. Intuitively, this is the maximum value that can be guaranteed if agent $i$ partitions the items into $d$ bundles and chooses the least valued bundle. 
When it is clear from the context, we write $\text{MMS}_{i}^{d}$ or $1$-out-of-$d$ MMS to refer to $\text{MMS}_{i}^{d}(M)$.

Given an instance, we say that a $1$-out-of-$d$ MMS exists if there exists an allocation  $A = (A_1, \ldots, A_n) \in \Pi_{n}(M)$ such that for every agent $i\in N$, $v_{i}(A_i) \geq \text{MMS}_{i}^{d}(M)$.
Note that $\text{MMS}_{i}^{d}(M) \leq \frac{v_i(M)}{d}$ and it is a weakly-increasing function of $d$: a larger $d$ value means that there are more agents to share the burden, so each agent potentially has fewer chores to do. Clearly, $\text{MMS}_{i}^{d} = \frac{v_i(M)}{d}$ when chores can be partitioned into $d$ bundles of equal value. Moreover,  $\frac{v_i(M)}{n}$ is agent $i$'s \textit{proportional share}.

\paragraph{\textbf{Ordered instance.}} An instance $I$ is \textit{ordered} when all agents agree on the linear ordering of the items, irrespective of their valuations. Formally, $I$ is an \emph{ordered instance} if there exists an ordering $(c_{1}, c_{2}, \ldots, c_{m})$ such that for all agents $i\in N$ we have 
$|v_{i}(c_{1})| \geq |v_{i}(c_{2})| \geq \ldots \geq |v_{i}(c_{m})|$. 
Throughout this paper, we often refer to this as an ordering from the \textit{largest chores} (least preferred) to the \textit{smallest chores} (most preferred).

In the context of allocating goods, \citet{Bouveret2016} introduced ordered instances as the `most challenging' instances in achieving MMS, and showed that given an \textit{un}ordered instance, it is always possible to generate a corresponding ordered instance in polynomial time.\footnote{\citet{Bouveret2016} called these same-order preferences.}
More importantly, if an ordered instance admits an MMS allocation, the original instance also admits an MMS allocation which can be computed in polynomial time (see \cref{ex:ordering}).


\begin{lemma}[\citet{barman2017approximation}]\label{lem:order}
Let $I' = \ins{N, M, V'}$ be an ordered instance constructed from the original instance $I = \ins{N, M, V}$. 
Given allocation $A'$ on $I'$, a corresponding allocation $A$ on $I$ can be computed in polynomial time such that for all $i\in N, v_{i}(A_{i}) \geq v'_{i}(A'_{i})$.
\end{lemma}

The above results hold for any MMS approximation without loss of generality, and have been adopted extensively in simplifying the MMS approximations of chores \cite{huang2021algorithmic}. Therefore, throughout the paper we only focus on ordered instances.

\begin{example}[Ordering an instance]\label{ex:ordering}
Consider the following unordered instance with four chores and two agents:

\begin{center}\small
\begin{tabular}{c| c c c c | c | c}
               & $c_{1}$   & $c_{2}$   & $c_{3}$   & $c_{4}$      & $\text{MMS}^{n}_{i}$ & $v_i(A_i)$  \\\hline
        $a_{1}$ & -3         & \circled{-5}         & -6         & \circled{-1}         & -8 & -6 \\ 
        $a_{2}$ & \circled{-2}       & -8         & \circled{-4}          & {-9}          & -12  & -6 \\
    \end{tabular}
\end{center}

An ordered instance is obtained by sorting the values in descending order of absolute values. It has two possible allocations marked by a circle and $^{*}$ that satisfy MMS:

\begin{center}\small
\begin{tabular}{c| c c c c | c}
                & $c'_{1}$   & $c'_{2}$   & $c'_{3}$   & $c'_{4}$      & $\text{MMS}^{n}_{i}$ \\\hline
        $a_{1}$ & -6$^{*}$            & \circled{-5}           &    \circled{-3}    &    -1$^{*}$        & -8 \\ 
        $a_{2}$ & \circled{-9}             & -8$^{*}$          &    {-4}$^{*}$     &    \circled{-2}     & -12 \\
    \end{tabular}
\end{center}
Any of the marked MMS allocations in the ordered instance corresponds to a picking-sequence that results in an MMS allocation in the original instance. 
A picking sequence lets agents select items from the `best chores' (most preferred) to the `worst chores' (least preferred).

For instance, applying a picking sequence 2, 1, 1, 2 (obtained from the circled allocation in the second table) to the original instance results in allocation $A$ (marked by circles in the first table) that guarantees MMS.
Specifically, when applied to the original instance, agent 2 picks first, and takes its highest valued chore $c_1$, which corresponds to $c'_4$. 
Agent 1 picks next. Since its best chore $c_4$ is available he picks it. 
The next pick also belongs to agent 1. But his second-best chore is $c_1$, which is already allocated to agent 2. Thus, agent 1 picks its next-best available chore $c_2$, and agent 2 is left with $c_3$.

\end{example}


\section{Valid Reductions for Chores} \label{sec:reductions}
In this section, we first show that the valid reductions techniques that are typically used for allocating goods can no longer be applied to chores instances.
While typical goods reductions fail in allocating chores, we then argue that some of the core ideas translate to chores allocation through careful adaptations. 
These techniques are of independent interest as they can be utilized in other heuristic algorithms (e.g. multiplicative MMS approximations).


\subsection{Reductions for goods}

Several algorithms that are developed to provide multiplicative MMS approximations rely on structural properties of MMS and heuristic techniques to avoid computational barriers of computing MMS thresholds.
To understand common reduction techniques, we first take a detour to recall techniques that are valid when allocating goods.
For the ease of exposition, we present this section with the standard definition of $1$-out-of-$n$ MMS. 






\begin{definition}[Valid Reduction for Goods]
Given an instance, $I = \ins{N, M, V}$ and a positive integer $n$, allocating a set of goods $A_i \subseteq M$ to an agent $i\in N$ is a \emph{valid reduction} if 

(i) $v_i(A_i) \geq \text{MMS}_{i}^{n}(M)$, and

(ii) $\forall j \in N \setminus \{i\}, \text{MMS}_{j}^{n-1}(M \setminus A_i) \geq \text{MMS}_{j}^{n}(M)$.
\end{definition}

Intuitively, a valid reduction ensures that the MMS values of the remaining agents in the reduced instance does not strictly decrease; otherwise, solving the reduced instance may violate the initial MMS values of agents. 

Since computing MMS values is NP-hard \citep{Bouveret2016}, one can instead utilize proportionality as a (loose) upper bound for MMS values. 
Given the proportionality bound, it is easy to see that for each agent $i\in N$, $\text{MMS}_{i}^{n}(M) \leq \frac{v_i(M)}{n}$. Therefore, any good $g\in M$ with a value $v_{i}(g) \geq \frac{v_i(M)}{n}$ for agent $i$ can be  assigned to agent $i$, satisfying $i$'s MMS value, without violating conditions of valid reductions.
The next lemma (due to \citet{garg2018approximating}) formalizes this observation and provides two simple reduction techniques.


\begin{lemma}[\citet{garg2018approximating}]\label{lem:n_goods_reduction}
Given an ordered goods instance $I = \ins{N, M, V}$ with $|N|=n$, if $v_{i}(\{g_{n},g_{n+1}\}) \geq \frac{v_i(M)}{n}$, then allocating $A_{i} = \{g_{n}, g_{n+1}\}$ to agent $i$ 
(and removing them from the instance) forms a valid reduction. Similarly, allocating $\{g_{1}\}$ to agent $i$ forms a valid reduction if $v_{i}(\{g_{1}\}) \geq \frac{v_i(M)}{n}$.
\end{lemma}

The following example illustrates how valid reductions can be iteratively applied to reduce an ordered instance.

\begin{example}[Valid reductions for goods]
Consider five goods and three agents with valuations as shown in the table below.
\begin{center}\small
\begin{tabular}{c| c c c c c | c}
                & $g_{1}$   & $g_{2}$   & $g_{3}$   & $g_{4}$   & $g_{5}$   & $\text{MMS}_{i}^{3}$ \\\hline
        $a_{1}$ & {9}         & {6}         & \circled{5}         & \circled{3}         & {1}         & 7 \\ 
        $a_{2}$ & 8         & \circled{7}         & {6}         & {2}         & \circled{1}         & 8 \\
        $a_{3}$ & \circled{10}        & 8         & 5         & 3         & 1         & 8 \\
    \end{tabular}
\end{center}
The MMS values of all three agents are shown in the table. Suppose $g_1$ is allocated to agent $a_3$. This allocation is a valid reduction because $v_{3}(g_1) \geq MMS^{3}_{3}(M)$. After this reduction, the MMS values for the remaining agents are $\MMS{2}{1}(M\setminus \{g_1\}) = 7$ and $\MMS{2}{2}(M\setminus \{g_1\}) = 8$ respectively. At this point, the set $\{g_3, g_4\}$ can be given to agent $a_1$ as a valid reduction since $g_3$ and $g_4$ are precisely $n$th and $(n+1)$th highest valued goods according to $a_1$ in the reduced instance (note that $n = 2$ after the removal of $a_3$).
\end{example}

\begin{remark}
When allocating goods, valid reduction techniques are often used together with scaling of an instance to simplify the approximation algorithms \cite{garg2020improved,garg2018approximating}.
The \textbf{scale invariance} property of MMS \citep{ghodsi2018fair} states that if an agent's valuations are scaled by a factor, then its MMS value scales by the same factor. 
Formally, given an instance $I = \ins{N, M, V}$, for every agent $i\in N$ with a proportionality bound $\frac{v_{i}(M)}{n}$ we can construct a new instance $I' = \ins{N, M, V'}$ such that $v'_{i}(M) = n$ and for every $g\in M$, $v'_{i}(g) = \frac{n}{v_{i}(M)} v_{i}(g)$.
Using the proportionality bound for scaling an instance implies that allocating any set $S\in M$ such that $v_{i}(S) \geq 1$ to agent $i$ forms a valid reduction.
\end{remark}

The scale invariance property of MMS and reduction techniques circumvent the exact computation of MMS thresholds, which enables greedy approximation algorithms for allocating goods. \citet{garg2018approximating} developed a simple greedy algorithm that guarantees to each agent $2/3$ of its MMS value; later algorithms improved this approximation to $3/4$ \citep{garg2020improved,ghodsi2018fair}.

\subsection{Failure of Goods Reductions} \label{sec:failgoodsreduction}
We briefly discuss how the valid reductions for goods do not translate to instances with chores. The reason is that the reductions for goods rely upon the fact that, redistributing items from one bundle of a partition to other bundles weakly increases the value of other bundles. However, in the context of chores, this assumption does not hold as we illustrate next.

\begin{example}
Consider three agents and six chores. Agents' valuations are identical such that each agent $i\in N$ values each chore $c\in M$ as $v_{i}(c) = -1$.
%
%
The 1-out-of-$3$ MMS of all agents is $-2$, i.e. $\MMS{3}{i} = -2$ for every $i\in N$. A reduction that allocates a single chore (e.g. largest chore), say $c_1$, satisfies agent 1 since $v_1(c_1) = -1 \geq \MMS{3}{1}$. However, this reduction is \textit{not} valid since the MMS value of the remaining agents decreases, that is, $\MMS{2}{i} = -3$ for $i\in\{2,3\}$.

\end{example}

To illustrate why reductions of larger bundles such as $\{c_n, c_{n+1}\}$ fail, we provide the following example that generalizes this reduction to  bundles with larger sizes.

\begin{example}
Consider an instance with three agents and $3(k+2)$ chores that are each valued $-1$.
Each agent's MMS value is $\MMS{3}{i} = -(k+2)$. Take any bundle $S \subset M$ of $k+1$ chores. Any agent $i$ would agree to receive $S$, as $v_{i}(S) = -(k+1) \geq \MMS{3}{i} = -(k+2)$. 
However, allocating the bundle $S$ to agent $i$ is not a valid reduction. This is because the remaining $2k+5$ chores must be allocated among the remaining two agents, but $\MMS{2}{j}(M\setminus S) = -(k+3)$ which is less than $\MMS{3}{j} = -(k+2)$. 

Notice that smaller bundles of $v_{i}(\{c_n, c_{n+1}\}) = -2$ do satisfy agent $i$ as well but still result in decrease of MMS values for other agents. For example, when $k = 2$, if $\{c_3, c_4\}$ are allocated to an agent, the MMS values of the remaining agents decrease from $\MMS{3}{i} = -4$ to $\MMS{2}{i} = -5$.
\end{example}

\subsection{Estimating the Number of Large Chores}
One of the key distinctions between allocating goods and chores is the tolerance of bounds used for approximating MMS values. 
As we discussed previously, proportionality provides a reasonable upper bound in allocating goods through reductions: as soon as the value of a bundle reaches an agent's proportionality threshold, a reduction can be applied without including any additional item.

In contrast, when allocating chores, proportionality may be a loose bound: when selecting a set of chores that satisfies proportionality for an agent, it may still be necessary to include additional chores to ensure that no chore remains unallocated.

\begin{example}
Consider an instance with 10 chores and 10 agents with identical valuations: three small chores valued at $-\frac{1}{3}$, six medium chores valued at $-\frac{1}{2}$, and one large chore valued at $-1$. 
The proportionality threshold is $-\frac{1}{2}$ but the MMS is $-1$.
Once an agent reaches the proportionality threshold, say by receiving a single medium chore, it could still receive an additional medium or small chore.
\end{example}

The main challenge is how to pack as many chores as possible within a bundle without violating the maximin share threshold.


We start by making a simple assumption on the size of the instance. For any instance, without loss of generality, we can always add dummy chores with value $0$ and assume that $m \geq 2 d +1$.\footnote{In Appendix~\ref{app:size} we show that this assumption is valid without adding dummy chores.}

Our first lemma will be used to bound the number of large chores in each bundle. It states that in an ordered chores instance, the most preferred $k+1$ chores from the set of the least preferred $kd+1$ chores are valued at least as much as $1$-out-of-$d$ MMS share.

\begin{lemma}\label{lem:counting}
Let $I = \ins{N, M, V}$ be an ordered chores instance, and $k$ and $d$ be non-negative integers such that $kd + 1 \leq m$. Then, for each agent $i \in N$, 
$$v_{i}(\{c_{kd - (k-1)}, c_{kd - (k-2)}, \ldots, c_{kd + 1}\}) \geq \text{MMS}_{i}^{d}(M).$$
\end{lemma}
\begin{proof}
Consider the subset of chores $S = \{c_{1}, c_{2}, \ldots, c_{kd+1}\}$. 
By definition, for every chore $c\in M$, $v_{i}(c) \leq 0$, thus we have $\text{MMS}_{i}^{d}(S) \geq \text{MMS}_{i}^{d}(M)$.
By the pigeonhole principle, since $|S| > kd$, any partition of $S$ into $d$ bundles $(A_{1}, \ldots, A_{d})$ must contain at least one bundle, say $A_{\ell}$, which contains at least $k+1$ chores. By definition, we have $v_{i}(A_\ell) \geq \text{MMS}_{i}^{d}(S)$.

Let the set $B\subset S$ contain the $k+1$ last (most preferred) chores of $S$. 
most preferred chores of $S$. 
Since chores are ordered from the least to the most preferred chores, this $B$ is weakly preferred to $A_\ell$. Thus, $v_{i}(A_{\ell}) \leq v_{i}(B)$ where $B = \{c_{kd - (k-1)}, c_{kd - (k-2)}, \ldots, c_{kd + 1}\}$. By transitivity, $v_{i}(B) \geq v_{i}(A_{\ell}) \geq \text{MMS}_{i}^{d}(S)$.
\end{proof}

\cref{lem:counting} links the number of chores to their values, and enables us to identify the number of large (least preferred) chores.

\begin{corollary} \label{cor:choreRed}
Given an ordered chores instance $I = \ins{N, M, V}$, and an integer $d\geq 1$, the following statements hold\footnote{If $d < 2d+1$, we may add $2d+1-m$ dummy chores with value $0$ to all agents.}:
\begin{enumerate}
    \item $v_i(\{c\}) \geq \MMS{d}{i}(M)$, for all $c\in M$;
    \item $v_i(\{c_d, c_{d+1}\}) \geq \MMS{d}{i}(M)$;
    \item $v_{i}(\{c_{2 d-1}, c_{2 d},~~ c_{2 d+1}\})\geq \MMS{d}{i}(M)$.
\end{enumerate}
\end{corollary}

\begin{proof}
By setting $k = 0$ in \cref{lem:counting}, for each agent $v_{i}(\{c_1\}) \geq \MMS{n}{i}(M)$. Since $c_1$ is the worst chore in an ordered instance, for every other chore $c\in M$, $v_i(\{c\}) \geq \MMS{d}{i}(M)$.
Similarly, setting $k = 1$ and $k=2$ in \cref{lem:counting} yields claims (2) and (3).
\end{proof}

\section{$1$-out-of-$\lfloor\frac{2n}{3}\rfloor$ Maximin Share for Chores in Polynomial Time}
\label{sec:Chores_TwoThirds}

In this section, we present a polynomial-time algorithm for allocating chores that achieves $1$-out-of-$\lfloor\frac{2n}{3}\rfloor$ MMS. 
The algorithm takes a chores instance along with a set of thresholds for agents as an input and utilizes a greedy ``bag-filling'' procedure to assign bundles of chores to agents.
The high-level idea behind the algorithm is allocating the large (least desirable) chores first and packing as many chores as possible into a bundle up to the given threshold. 
The algorithmic idea is simple. The key in achieving $1$-out-of-$\lfloor\frac{2n}{3}\rfloor$ MMS approximation is selecting appropriate threshold values.



\paragraph{Algorithm description.}
The underlying structure of Algorithm~\ref{alg:chores_bagfill} is 
similar to the 
\textit{First-Fit-Decreasing} algorithm for bin-packing
\citep{johnson1973near}.%
\footnote{
The same algorithm is used by  \citet{huang2021algorithmic} for achieving multiplicative approximations of MMS.
They prove that, with appropriate thresholds, Algorithm \ref{alg:chores_bagfill} guarantees every agent at least $11/9$ of its MMS value. This does not directly imply any result for ordinal approximation as shown in Example \ref{exm:ordinal-vs-cardinal}.
}
It starts by selecting an empty bundle and adding a large (lowest value) chore  to the bag. While the value of the bag is above a threshold for at least one agent, add an additional chore---in order of the largest to smallest---to the bundle. 
If a chore cannot be added, the algorithm skips it and considers the next-smallest (more preferred) chores. Each agent has a different threshold, $\beta_{i}$, and assesses the bundle based on this threshold. 
When no more chores can be added, the bundle is allocated to an arbitrary agent who still finds it acceptable. The algorithm repeats with the remaining agents and chores.

\begin{algorithm}[t]
\small
\caption{\small{Algorithm for $1$-out-of-$d$ MMS approximation}}
\label{alg:chores_bagfill}
\KwIn{An ordered chores instance $I=\ins{N, M, V}$ and threshold values $(\beta_{i})_{i=1}^{n}$ with $\beta_i\leq 0$ for all $i\in N$.}
\KwOut{Allocation $A = (A_1,\ldots,A_n)$ satisfying $v_{i}(A_{i}) \geq \beta_{i}$ for all $i\in N$.}
\While(\tcp*[h]{there are remaining agents}){$|N| > 0$}{
\Comment{Adding as many chores as possible to a bundle}
Initialize $B$ as an empty bundle \;
\For{each remaining chore $c$ in descending order of absolute values (hardest to easiest chore)}{
\If{there exists agent $i$ s.t. $v_i(\{B \cup c\}) \geq \beta_i$}{
$B \leftarrow B \cup \{c\}$ \tcp{Adding $c$ to $B$} 
}
}
\Comment{Allocating the bundle to an agent.}
Select an agent $i$ such that $v_i(B) \geq \beta_i$ (arbitrary break ties)\;
$A_i \leftarrow B$ \;
$N \leftarrow N \setminus \{i\}$ \;
$M\leftarrow M \setminus B$ \;
}
\end{algorithm}

For any selection of non-positive thresholds $(\beta_{i})_{i=1}^{n}$, 
Algorithm~\ref{alg:chores_bagfill} guarantees that 1) every bundle is allocated to an agent who values it at least $\beta_{i}$, and 2) every agent receives a bundle (possibly an empty bundle).  However, if the thresholds are too optimistic (too close to zero), the algorithm may result in a partial allocation, i.e., some chores might remain unallocated. The main challenge is to carefully choose the threshold values such that the algorithm will provably terminate with a \textit{complete} allocation.%
\footnote{In contrast, when allocating goods, all goods are allocated, and the challenge is showing that all \emph{agents} receive a bundle of certain threshold.}

\begin{theorem}\label{thm:Chores_TwoThirds}
Given an additive chores instance, a $1$-out-of-$\floor{\frac{2n}{3}}$ MMS allocation exists and can be computed in polynomial time.
\end{theorem}

\begin{proof}
Let $I = \ins{N, M, V}$ be an ordered instance and $d = \floor{\frac{2n}{3}}$. 
Without loss of generality, we can assume that $m\geq 2d+1$ by adding dummy chores with value 0 for all agents.

For each agent, let the thresholds be selected as follows:
\begin{equation*}
   \beta_{i} = \min\left(v_{i}(c_{1}),~~ v_{i}(\{c_{d}, c_{d+1}\}),~~ v_{i}(\{c_{2 d-1}, c_{2 d},~~ c_{2 d+1}\}), \frac{v_{i}(M)}{d}\right). 
\end{equation*}
\cref{cor:choreRed} and the inequality $\frac{v_{i}(M)}{d}\geq \text{MMS}_{i}^{d}(M)$ imply that all agents receive their $1$-out-of-$d$ MMS, that is, $\beta_{i} \geq \text{MMS}_{i}^{d}(M)$.

In order to show that all chores are allocated, we split the chores into three categories of large ($\{c_{1}, \ldots, c_{d}\}$), medium $(\{c_{d+1}, \ldots, c_{2d}\}$), and small $(\{ c_{2d+1}, \ldots, c_m\})$ chores.

Since for all $i\in N$, $v_{i}(c_{1}) \geq \beta_{i}$, every single chore can be added to an empty bag. 
Consider the first $d$ bundles. Since these bundles contain at least one chore each, and $d\leq n$, the $d$ large chores are allocated within the first $d$ iterations.

Similarly, since $v_{i}(c_{d},c_{d+1}) \geq \beta_{i}$, the medium chores may be bundled in pairs from largest to smallest and form the next bundles.
This implies that, within the first $d + \ceil{\frac{d}{2}}$ allocated bundles, all large and medium chores are allocated. Importantly, 
$$d + \ceil{\frac{d}{2}} = \floor{\frac{2n}{3}} + \ceil{\frac{\floor{\frac{2n}{3}}}{2}} \leq \floor{\frac{2n}{3}} + \ceil{\frac{n}{3}} = n.$$

Thus, we conclude that all large and medium chores are allocated upon the termination of the algorithm.

The last step is to prove that all small chores are allocated too. These chores are added to bundles whenever there is additional gap between $v_{i}(A_{i})$ and $\beta_{i}$. 
Consider the last agent, $i$, who receives a bundle before Algorithm~\ref{alg:chores_bagfill} terminates.
If no small chores remain before agent $i$ receives a bundle, then we are done.

Suppose that there is some remaining small chore $c$ before agent $i$ receives a bundle. For each other bundle $A_{j}$ already allocated, necessarily $v_{i}(A_{j} \cup \{c\}) < \beta_{i}$,  because otherwise agent $i$ would have accepted $A_{j} \cup \{c\}$ and chore $c$ would have been added to $A_{j}$. Now, since $v_{i}(\{c_{2d-1}, c_{2d}, c_{2d+1}\}) \geq \beta_{i}$ and the instance is ordered, we have that $v_{i}(c) \geq v_{i}(c_{2d+1}) \geq \frac{\beta_{i}}{3}$. In turn, this implies that $v_{i}(A_{j}) < \beta_{i} - v_{i}(c) = \frac{2\beta_{i}}{3}$ for each $j \neq i$. 

By the way we selected the thresholds, we have that $\beta_{i} \leq \frac{v_{i}(M)}{d}$. 
We use this fact to upper bound the amount of value in each previously allocated bundle: 
$$v_{i}(A_{j}) < \frac{2\beta_{i}}{3} ,$$

which implies that
$$v_{i}(A_{j}) < \frac{2v_{i}(M)}{3d}.$$

By replacing the value of $d$, we have
$$v_{i}(A_{j}) < \frac{2v_{i}(M)}{3\floor{\frac{2n}{3}}}.$$

Therefore,

$$v_{i}(A_{j}) \leq \frac{2}{3} \cdot \frac{3}{2} \cdot \frac{v_{i}(M)}{n} = \frac{v_{i}(M)}{n}.$$

This inequality implies that before the last bundle is initialized, agent $i$ values the remaining items at least $v_{i}(M) - \sum_{j \neq i} v_{i}(A_{j}) > v_{i}(M) - (n-1) \frac{v_{i}(M)}{n} = \frac{v_{i}(M)}{n} \geq \beta_{i}$. Thus, agent $i$ can take all the remaining chores.
\end{proof}

\begin{remark}
Interestingly, for goods, $1$-out-of-$\floor{\frac{3n}{2}}$ MMS approximations exist \cite{hosseini2021mms} and can be computed in polynomial time \cite{hosseini2021ordinal}. However, the techniques used for proving the existence results as well as developing a tractable algorithm are substantially different due to reductions available for goods (as discussed in \cref{sec:reductions}) as well as challenges posed by packing bundles as much as possible to ensure complete allocations of chores.
On the other hand, in the case of goods even a slight error in computing MMS values may result in wasting values and not having sufficient goods to satisfy some agents (see \citep{hosseini2021mms} for an example) whereas for chores we can tolerate an estimate of MMS values as long as all chores are  allocated.
\end{remark}


\section{$1$-out-of-$\lfloor\frac{3n}{4}\rfloor$ MMS Allocations Exist for Chores}
\label{sec:Chores_ThreeFourths}

In this section, we show that a careful selection of threshold values in Algorithm~\ref{alg:chores_bagfill}, in fact, guarantees $1$-out-of-$\floor{\frac{3n}{4}}$ MMS approximation. To achieve this result we require a precise computation of MMS values for each agent, which in turn is intractable \citep{Bouveret2016}. Nonetheless, we prove the existence of $1$-out-of-$\lfloor\frac{3n}{4}\rfloor$ MMS, and later in \cref{sec:polyApprox} provide a polynomial-time algorithm that achieves an approximation of this bound.

\begin{theorem}
\label{thm:Chores_ThreeFourths}
Given an additive chores instance, a $1$-out-of-$\floor{\frac{3n}{4}}$ MMS allocation is guaranteed to exist.
\end{theorem}

\Cref{thm:Chores_ThreeFourths} is an immediate corollary of \cref{lem:Chores_ThreeFourths} below. For the ease of exposition, we first provide the proof of the theorem.

\begin{proof}
By construction, Algorithm~\ref{alg:chores_bagfill} terminates and every agent $i\in N$ receives a bundle (possibly empty) with the value of at least $\beta_i$. 
By Lemma \ref{lem:Chores_ThreeFourths}, we can pick for each agent $i$ the threshold $\beta_i = \MMS{d}{i}(M)$ where $d = \floor{\frac{3n}{4}}$, and all chores will be allocated. Thus, we have a complete allocation in which each agent's value is at least  $1$-out-of-$\floor{\frac{3n}{4}}$ MMS, which proves \Cref{thm:Chores_ThreeFourths}.
\end{proof}

\begin{lemma}
\label{lem:Chores_ThreeFourths}
Suppose Algorithm~\ref{alg:chores_bagfill} is executed with threshold values $\beta_i \leq \MMS{\floor{\frac{3n}{4}}}{i}(M)$ for all $i\in N$. Then all chores are allocated upon termination of the algorithm.
\end{lemma}

\begin{proof}
Let $I = \ins{N, M, V}$ be an ordered chores instance. 
For simplicity, we start by scaling the valuations such that for each agent $i\in N$, $\MMS{\floor{\frac{3n}{4}}}{i}(M) = -1$.\footnote{This scaling step is only used to simplify the proof. An identical result can be achieved without scaling the valuations by setting all thresholds to $\beta_{i} = \MMS{d}{i}(M)$ where $d = \floor{\frac{3n}{4}}$ and updating the rest of the values in the proof accordingly.} 
This implies that 
\begin{align}
\label{eq:3n/4}
v_{i}(M) \geq -\floor{\frac{3n}{4}}\geq -\frac{3n}{4}
\end{align}
and $\beta_i\leq -1$ for each agent $i \in N$. 


Let agent $i$ be the last agent who received a bundle (in the $n$-th iteration). The proof proceeds by considering two types of remaining chores according to their value: 1) small chores $c\in M$ with value $v_i(c) \geq -\frac{1}{4}$, and 2) large chores $c\in M$ with value $v_i(c) < -\frac{1}{4}$.

\textbf{Case 1: small chores.} 
Suppose for contradiction that there is some chore $c\in M$ such that $v_{i}(c) \geq -\frac{1}{4}$ that remains unallocated at the end of the algorithm. By assumption, agent $i$ could not add $c$ to any allocated bundle, including $i$'s own bundle. Since $i$ is the last agent, we infer that for each agent $j\in N$ with bundle $A_j$, $v_{i}(A_{j} \cup \{c\}) < -1$.
By additivity, because $v_{i}(c) \geq -\frac{1}{4}$, we can write $v_{i}(A_{j}) < -\frac{3}{4}$ for all $j\in N$.
Summing over all assigned bundles gives $v_i(M) < -\frac{3n}{4}$, which contradicts \eqref{eq:3n/4}. Therefore, no such small chore remains at the end of the algorithm.

\textbf{Case 2: large chores.}
Suppose that there is some chore $c\in M$ such that $v_{i}(c) < -\frac{1}{4}$ that remains unallocated at the end of the algorithm. We define the following sets of bundles.
\begin{itemize}
\item 
 $M_{1}, \ldots, M_{\floor{\frac{3n}{4}}}$ are \emph{MMS bundles} --- bundles that comprise a $\MMS{\floor{\frac{3n}{4}}}{i}(M)$ partition of agent $i$.
\item $B_1,\ldots,B_n$ are \emph{algorithm bundles} --- bundles allocated by Algorithm \ref{alg:chores_bagfill}. $B_t$ denotes the bundle allocated at iteration $t$.
\end{itemize}
For each MMS bundle $M_j$, let $M_j[s]$ denote the $s$-th largest chore (least valued) of $M_{j}$.
Whenever $|M_{j}| < s$, we define $v_{i}(M_{j}[s]) = 0$.
Without loss of generality, we assume that the MMS bundles are sorted such that $|v_{i}(M_{1}[1])| \geq |v_{i}(M_{2}[1])| \geq \ldots \geq |v_{i}(M_{\floor{\frac{3n}{4}}}[1])|$. 
Since valuations are scaled so that  $\MMS{\floor{\frac{3n}{4}}}{i}(M) = -1$,
there are at most $3$ large chores (with value less than $\frac{1}{4}$) in each MMS bundle.

For the sake of the proof, we maintain a vector of \emph{shadow-bundles} $M_{1}',
 M_{2}', \ldots, M_{n}'$, which is initialized as follows:
\begin{itemize}
\item For each $j\in\{1,\ldots,\floor{\frac{3n}{4}}\}$, $M_{j}' \coloneq $ the set of large chores (with value less than $-\frac{1}{4}$ to $i$) in $M_{j}$.
\item For each $j\in\{\floor{\frac{3n}{4}},\ldots,n\}$,
$M_{j}' \coloneq \emptyset$.
\end{itemize}
At each iteration $t$ of the algorithm, we edit the vector of shadow-bundles by moving some chores between bundles. We do so such that, at the start of iteration $t$, the following invariants hold:
\begin{enumerate}
\item $M'_j\subseteq B_j$ for all $j<t$. That is, each chore in the  shadow-bundles $M_{1}', \ldots, M_{t-1}'$ is allocated.
\item $|M_{j}'| \leq 3$ and $v_{i}(M_{j}') \geq -1$ for $j \geq t$.
That is, each remaining shadow-bundle $M_{t}', \ldots, M_{n}'$ has value at least $-1$.
\end{enumerate}
Both invariants hold before the first iteration ($t=1$):
invariant (1) holds vacuously, and invariant (2) holds since each bundle $M_{j}'$ is contained in one of $i$'s MMS bundles.

Suppose the invariants hold before iteration $t\geq 1$. We show how to edit the shadow-bundles such that the invariants still hold before iteration $t+1$.

We reorder the shadow-bundles  $M_{t}', \ldots, M_{n}'$ so that $M_{t}'[1]$ is the largest remaining chore. Hence, in iteration $t$, Algorithm~\ref{alg:chores_bagfill} selects this chore first to add to the bag.
That is, $B_t[1] = M_{t}'[1]$.
We split to cases based on the size of $|M_{t}'|$, which must be in $\{1,2,3\}$ 
by invariant (2).

If $|M_{t}'| = 1$,
then both invariants hold at $t+1$, since $M_{t}' \subseteq B_{t}$, and the shadow-bundles do not change.

If $|M_{t}'| = 2$,
then we have to handle $M_{t}'[2]$.
By invariant (2) we have 
$M_{t}'[1]+M_{t}'[2]\geq -1$. This means that $M_{t}'[2]$ can potentially be inserted as the second chore in $B_t$. If indeed $B_t[2]=M_{t}'[2]$, then we are done --- both invariants hold at $t+1$, since $M_{t}' \subseteq B_{t}$, and the shadow-bundles do not change.
If $B_t[2]\neq M_{t}'[2]$, 
this means that 
Algorithm~\ref{alg:chores_bagfill} processed chore $B_t[2]$ before chore $M_{t}'[2]$. 
Since the algorithm processes jobs by ascending order of values (descending order of absolute values), this implies that 
$v_i(B_t[2])\leq v_i(M_{t}'[2])$.
Now, we find the chore $B_t[2]$ in some shadow-bundle $M'_j$ for some $j>t$, and swap it with $M_{t}'[2]$. We claim that both invariants still hold:
\begin{enumerate}
\item $M'_t\subseteq B_t$, since after the swap $B_t[1]=M'_t[1]$ and $B_t[2]=M'_t[2]$, and $|M'_t|=2$.
\item The remaining shadow bundles remained as before, except for the shadow-bundle $M'_j$, in which a single chore was swapped. But, because $v_i(B_t[2])\leq v_i(M_{t}'[2])$, the value of $M'_j$ weakly increases, so it is still at least $-1$.
\end{enumerate}

Finally, suppose $|M_{t}'| = 3$.
We handle $M_t'[2]$ as in the previous case, so that now
$M'_t[1]=B_t[1]$
and
$M'_t[2]=B_t[2]$.
It remains to handle $M_{t}'[3]$.
Because $M_{t}'[3]$ is the smallest chore in $M_{t}'$, and
$v_i(M_t')\geq -1$,
by the pigeonhole principle we must have $v_{i}(M_{t}'[3]) \geq -\frac{1}{3}$. We move chore $M_{t}'[3]$ to a bundle $M_{j}'$ which was initially empty and which contains fewer than $3$ chores (all of which were moved to the bundle this way and thus have value at least $-\frac{1}{3}$). Such a bundle can always be found because at most one chore is moved this way in each iteration, and there are at most $\floor{\frac{3n}{4}}$ bundles $M_{j}'$ which were initially non-empty. Thus an upper bound on the number of bundles filled this way is: $\ceil{\frac{1}{3} \cdot \floor{\frac{3n}{4}}} \leq \frac{n}{4} \leq n - \floor{\frac{3n}{4}}$. Since each chore moved this way has value at least $-\frac{1}{3}$, we preserve invariant (2)
$|M_j'|\leq 3$ and $v_{i}(M_{j}') \geq -1$.
After the move, $M'_t$ contains only two chores, both of which are in $B_t$, so invariant (1) holds too.

We note that if the first chore $B_{t}[1]$ is selected from one of these growing bundles, then because this chore has value at least $-\frac{1}{3}$ and because chores are only moved if $v_{i}(M_{t}'[1]) < -\frac{1}{3}$, no more chores will be moved in later iterations.


The final step in proving the lemma is to move all chores from $B_{t} \setminus M_{t}'$ to $M_{t}'$. This step is necessary in order to guarantee that the largest remaining chore in later steps is not from $B_{t} \setminus M_{t}'$ (and thus $M_{t+1}'[1] \notin B_{t} \setminus M_{t}'$).\footnote{For example, consider $M_{1}' = \{c_{1}, c_{2}\}$ and $M_{2}' = \{c_{3}, c_{4}, c_{5}\}$. It is possible that $B_{1} = \{c_{1}, c_{2}, c_{3}\}$, which means that $B_{2}[1] = c_{4}$ but $M_{2}'[1] = c_{3}$.} We may do this because it preserves $M_{t}' \subseteq B_{t}$. Notice that $v_{j}(B_{t}) \geq -1$ for the agent $j \in N$ who received bundle $B_{t}$; however, we do not require that $v_{i}(B_{t}) \geq -1$, as agent $i$ is not be allocated the bundle $B_{t}$. 
Observe that the chores $B_{t} \setminus M_{t}'$ correspond to additional large chores which could be added to the bundle $M_{t}'$, and thus, in moving these chores, the value of bundles $M_{j}'$ for $j > t$ can only weakly increase and will remain at least $-1$.

Lastly, invariant (2) implies that after iteration $(n-1)$, $M_{n}'$ has value at least $-1$ for agent $i$. All remaining large chores lie in this bundle. Thus agent $i$ may take all such large chores. This implies that $M_{n}' \subseteq B_{n}$ and that no large chores remain when the algorithm terminates.
\end{proof}

We do not know whether the $\floor{3n/4}$ factor is tight in general.
The following proposition shows a non-tight upper bound on the performance of Algorithm \ref{alg:chores_bagfill} for large values of $n$.

\begin{proposition}[Upper bound for Algorithm \ref{alg:chores_bagfill}]
\label{prop:upper-bound-n}
For any integer $k\geq 0$,  there is an instance with $n=11k+7$ agents in which Algorithm \ref{alg:chores_bagfill} cannot guarantee to each agent its  $1$-out-of-$(9k+6)$ MMS.
\end{proposition}
\begin{proof}
When all agents have the same valuation and the same threshold, 
Algorithm \ref{alg:chores_bagfill} reduces to an algorithm for bin-packing known as \emph{First Fit Decreasing} (\emph{FFD}) \citep{johnson1973near,baker1985new}.
FFD sorts the chores by descending value, and allocates each chore to the first (smallest-index) agent who can take it without going over the threshold.
Algorithm \ref{alg:chores_bagfill} (with identical valuations and thresholds) does exactly the same, only in a different order: instead of making a single pass over all the chores and filling all bins simultaneously, 
it makes $n$ passes over the chores, and fills each bin in turn with the chores that would be inserted to it in that single pass.

\citet{dosa2007tight} and \citet{dosa2013tight} have shown that, for every integer $k\geq 1$, there is a bin packing instance in which the optimal packing needs $9k+6$ bins but FFD needs $11k+8$ bins.
We construct a chore allocation instance with $n=11k+7$ agents with identical valuations, taken from that bin-packing instance. Assume 
 that the agents' thresholds are at least their $1$-out-of-$(9k+6)$ MMS.
Then, after Algorithm \ref{alg:chores_bagfill} allocates bundles to all $n$ agents, some chores may remain unallocated.
\end{proof}

Consider \cref{prop:upper-bound-n} with $k=0$ and $n = 7$. By \cref{thm:Chores_ThreeFourths}, our algorithm achieves $\floor{3n/4} = 5$ ordinal approximation. This bound is tight since we cannot guarantee to all agents their 1-out-of-$6$ MMS. 
We present this tight example below.


\begin{example}[A tight example for Algorithm \ref{alg:chores_bagfill}]
\label{prop:upper-bound-7}
Consider an instance with $n=7$ agents and $m=20$ chores valued as follows for all agents: 
four chores valued at $-201$, four chores valued at $-102$, four chores valued at $-101$, and eight chores valued at $-98$. 
For each agent, the $1$-out-of-$6$ MMS partition contains the following bundles with the MMS value of $-400$:
\begin{itemize}
\item $4$ bundles of chores with values $\{-201,-101,-98\}$;
\item $2$ bundles of chores with values $\{-102,-102,-98,-98\}$.
\end{itemize}

With the threshold values set as -400, Algorithm \ref{alg:chores_bagfill} generates the following bundles:
\begin{itemize}
\item 4 bundles with chores $\{-201,-102\}$;
\item 1 bundle  with chores $\{-101,-101,-101\}$;
\item 1 bundle  with chores $\{-101,-98,-98,-98\}$;
\item 1 bundle  with chores $\{-98,-98,-98,-98\}$.
\end{itemize}
After allocating these 7 bundles, a chore with the value of $-98$ remains unallocated and cannot be added to any of the above bundles since it would violate the threshold of $-400$.
\end{example}

\section{Polynomial-time Approximations} \label{sec:polyApprox}
In this section, we develop an efficient approximation algorithm that achieves $1$-out-of-$\floor{\floor{\frac{3 n}{4}} - O(\log{n})}$ MMS for any chores instance. We rely on Algorithm~\ref{alg:chores_bagfill} while utilizing an efficient approximation algorithm to find reasonable threshold values. 


This result provides an interesting computational contrast between multiplicative and ordinal approximations of MMS for allocating chores: multiplicative approximations require exact MMS values, which can be seen as a \textit{job scheduling} problem where the goal is to minimize the makespan (the maximum completion time of a machine). However, ordinal MMS approximation on chore instances can be modeled as a combinatorial problem of \emph{bin packing} (see \citet{korte2018bin} for a detailed survey) where the goal is to minimize the number of bins subject to an upper bound on the total size of items in each bin.

While both problems are NP-hard, they differ in the approximation algorithms available for them. The job scheduling problem has polynomial-time approximation schemes (PTAS) \citep{woeginger1997polynomial}, but their runtime is exponential in the approximation accuracy $1/\epsilon$. 
On the other hand, the bin packing problem used for our ordinal MMS approximation admits \emph{additive} approximation algorithms. 

In particular, we use an algorithm by \citet{hoberg2017logarithmic}, which we call \emph{Algorithm HR}.
Algorithm HR takes as input a bin-packing instance $I$,
and returns a packing with at most $\lceil OPT(I)+a\cdot\log(OPT(I))\rceil$ bins (for some fixed constant $a$) in time polynomial in $m$ (the number of input numbers in $I$), where $OPT(I)$ denotes the smallest possible number of bins for $I$.
We combine Algorithm HR with binary search on the bin size.\footnote{Similar search techniques have been used for \textit{MultiFit} scheduling algorithms \citep{coffman1978application} and the dual approximation scheme of \citet{hochbaum1987using}.}

To efficiently apply binary search, we assume in this section that the values of chores are negative \textit{integers} with a bounded binary representation. The run-time of our algorithm will be polynomial in the size of the binary representation of the input.

\begin{algorithm}[t]
\small
\caption{\small{Computing an approximate MMS value}}
\label{alg:binary-search}
\KwIn{
An integer $d\geq 1$;
a single agent with value function $v_i$ over a set of chores $M$; all values are negative integers.
}
\KwOut{
A number  $\beta_i$ in the interval $[\MMS{\floor{d-\log{d}}}{i}(M) , \MMS{d}{i}(M)]$.
}
\Comment{Construct a bin-packing instance:}
Let $S := \big\{ - v_i(c) ~~ \big| ~~ c\in M \big\}$ \;
\Comment{Initialize a lower and an upper bound for the bin size:}
Let $L := 0$ \;
Let $U := (\sum S)$ rounded up to the nearest power of $2$\;
\Comment{Run binary search:}
\While{$U > L+1$}
{
Let $b := (U+L)/2$ \;
Run Algorithm HR \citep{hoberg2017logarithmic} on instance $S$ with bin-size $b$  \;
\If{at most $d$ bins are used}{
Let $U := b$; \hskip 1cm \tcp{Try smaller bins}
}
\Else{
Let $L := b$; \hskip 1cm   \tcp{Try larger bins}
}
}
\Return{$-U$.}
\end{algorithm}

\begin{lemma}
\label{lem:binary-search}
Given an additive chores instance with integer values, for any integer $d\geq 1$ and agent $i$, it is possible to compute a number $\beta_i$ for which 
\begin{align*}
\MMS{\floor{d-a\cdot \log{d}}}{i}(M) 
\leq
\beta_i
\leq
 \MMS{d}{i}(M) \leq 0,
\end{align*}
in time polynomial in the size of binary representation of the input.
\end{lemma}

\begin{proof}
We start by applying Algorithm \ref{alg:binary-search}.
The algorithm converts the chores allocation instance to a bin-packing instance, where each chore $c\in M$ is converted to an input of size $|v_i(c)|$.
Then it applies binary search with lower bound $L$ and upper bound $U$.
Throughout the search, the following invariants are maintained:
\begin{enumerate}
\item $U > L \geq 0$;
\item 
\label{inv:U}
Algorithm HR with bin-size $U$ needs at most $d$ bins;
\item 
\label{inv:L}
Algorithm HR with bin-size $L$ needs more than $d$ bins.
\end{enumerate}
The invariants are obviously true at initialization, and they are maintained by the way $U$ and $L$ are updated.
Let $\beta_i$ be the returned value, that is, the value of $-U$ once the algorithm terminates.
By the termination condition, at this point $U=L+1$.

Invariant \eqref{inv:U} implies that 
there exists a partition of chores into $d$ bins, in which 
the total absolute value of each bin is at most $U$, 
so the total value is at least $-U$. 
Therefore, $\MMS{d}{i}(M) \geq -U = \beta_i$.

Invariant \eqref{inv:L} implies that
there is no partition of the chores into $\floor{d-a\cdot\log{d}}$ or fewer bins, in which the total absolute value of all bins is at most $L$---otherwise the HR algorithm could have filled at most $\ceil{\floor{d-a\cdot\log{d}}+a\cdot\log\floor{d-a\cdot\log{d}}}\leq d$ bins of size $L$.
Therefore, 
$\MMS{\floor{d-a\cdot\log{d}}}{i}(M) < -L$. 
Since we assumed that all chores' values are integers, 
this implies 
$\MMS{\floor{d-a\cdot\log{d}}}{i}(M) \leq
-L-1 = -U = \beta_i$. 

The binary search uses $\lceil \log_2(\sum S)\rceil$ iterations, which is polynomial in the size of the binary representation of the input. Each iteration runs the HR algorithm, whose run-time is polynomial in $m$. 
This concludes the proof of the lemma.
\end{proof}

\begin{theorem}
\label{thm:chores-approx}
Given an additive chores instance with integer values,
it is possible to find in polynomial time, for some fixed positive constant $a$,
a
$1$-out-of-$\floor{\floor{\frac{3 n}{4}} -a\cdot\log{\floor{\frac{3 n}{4}}}}$ MMS allocation
\end{theorem}

\begin{proof}
We use Algorithm \ref{alg:chores_polytime} which starts by computing a threshold value $\beta_i$ for each agent $i\in N$ using Algorithm \ref{alg:binary-search}.
Then, it applies Algorithm \ref{alg:chores_bagfill} with the resulting thresholds for allocating the chores.

Lemma \ref{lem:binary-search} implies that $\beta_i \leq \MMS{d}{i}(M)$ with $d = \floor{\frac{3n}{4}}$ for all $i\in N$.
By \Cref{lem:Chores_ThreeFourths}, 
this implies that Algorithm \ref{alg:chores_bagfill} allocates all the chores.
Therefore, Algorithm \ref{alg:chores_bagfill} yields a complete allocation in which the value of each agent $i$ is at least $\beta_i$. By Lemma \ref{lem:binary-search}, this value is at least $\MMS{\floor{\floor{\frac{3n}{4}}-a\cdot\log{\floor{\frac{3n}{4}}}}}{i}(M)$, concluding the proof.
\end{proof}

\begin{algorithm}[t]
\small
\caption{\small{Algorithm for ordinal MMS approximation
in polynomial time}}
\label{alg:chores_polytime}
\KwIn{An ordered chores instance $I=\ins{N, M, V}$.}
\KwOut{Allocation $A = (A_1,\ldots,A_n)$ satisfying $v_{i}(A_{i}) \geq \MMS{d}{i}(M)$ for all $i\in N$, such that $d = \floor{\floor{\frac{3n}{4}}-a\cdot\log{\floor{\frac{3n}{4}}}}$}
\For{each agent $i\in N$:}{
Run Algorithm \ref{alg:binary-search}
with $d = \floor{\frac{3n}{4}}$ and valuation $v_i$\;
Let $\beta_i$ be the returned value\;
}
Run Algorithm \ref{alg:chores_bagfill} on $I$ with the threshold values $(\beta_1,\ldots,\beta_n)$.
\end{algorithm}


\section{Discussion}
\cref{thm:Chores_ThreeFourths} shows that, asymptotically (when $n$ is large),  Algorithm \ref{alg:chores_bagfill} guarantees 1-out-of-$(\approx 0.75 n)$ MMS. \cref{prop:upper-bound-n}, however, shows that this bound cannot be improved to 1-out-of-$(\approx 0.81 n)$ MMS using this algorithm. An immediate, but challenging, research direction is closing this approximation gap and developing polynomial-time algorithms beyond those presented in this paper. 

All of our results use the same algorithm (Algorithm \ref{alg:chores_bagfill}) to allocate the chores, but with different threshold values.
This approach is ``pluralistic'' in that it allows each agent to choose between these thresholds: each agent may choose whether to settle for a lower but easy-to-compute threshold of  Section \ref{sec:Chores_TwoThirds}, or put an extra effort to compute a higher thresholds of Sections \ref{sec:Chores_ThreeFourths}
or \ref{sec:polyApprox}. This pluralistic approach may be useful in other fair division settings.

\begin{acks}
Hadi Hosseini acknowledges support from NSF IIS grants \#2052488 and \#2107173.
Erel Segal-Halevi is supported by the ISF grant 712/20.
We are grateful to anonymous referees of EC 2021 and AAMAS 2022 for their valuable feedback.
\end{acks}




\bibliographystyle{ACM-Reference-Format} 
\bibliography{\ifnum\pdfstrcmp{\jobname}{output}=0 ref\else ../ref\fi}


\appendix
\begin{center}\huge
   \textbf{Appendix}
\end{center}

\section{Robustness of ordinal MMS approximations} \label{app:robust}
Consider an instance with $n=4$ agents and four chores $\{c_1,c_2,c_3,c_4\}$. Suppose that some agent values the chores at $-1,-5,-7,-9$ respectively.
We compare two MMS approximations:
\begin{itemize}
\item Cardinal: $4/3$ of the $1$-out-of-$n$ MMS (the guarantee of \cite{barman2017approximation}).  
\item Ordinal: $1$-out-of-$\floor{\frac{3n}{4}}$ MMS. 
\end{itemize} 
Note that the two approximations are comparable, since both of them can be seen as approximating $\frac{4}{3n}$ of the total value.
However, the cardinal approximation is sensitive to small changes in the chore values. For the given instance, the $1$-out-of-$n$ MMS is $-9$, so the cardinal approximation guarantee is $-12$; an algorithm with this guarantee is allowed to give our agent the bundle $\{c_2,c_3\}$.
But if the value of $c_4$ changes to $-9+3\epsilon$, then the cardinal approximation guarantee changes to $-12+4\epsilon$, so the bundle $\{c_2,c_3\}$ no longer satisfies the guarantee, and the algorithm must give our agent a better bundle. Thus, the validity of  a bundle is affected by infinitesimally-small change in an \emph{irrelevant} chore.

In contrast, the ordinal approximation guarantee in the given instance is $1$-out-of-$3$ MMS, which is $-9$. An algorithm with this guarantee can give our agent either the bundle $\{c_4\}$ or any better bundle. Any change in a chore value---as long as it does not change the order between the bundle values---does not affect the validity of the bundle $\{c_4\}$.

\section{Relations between approximate fairness notions}
\label{app:relations}
Relations between various approximate fairness notions were studied by \citet{amanatidis2018comparing} for goods. As far as we know, such relations were not studied for chores yet.
The following propositions show that, in general, the various fairness notions are independent.

\begin{table*}[t]
\begin{tabular}{|c|c|c|c|c|}
\hline
This $\downarrow$ implies $\rightarrow$ 
& \textbf{Ordinal MMS} & \textbf{Multiplicative MMS} & \textbf{EF1} & \textbf{EFx} \\
\hline
\textbf{Ordinal MMS} 
& - 
& at most $2$ \eqref{prop:ordinal-implies-cardinal} 
& None \eqref{prop:mms-no-ef1}
&  None \eqref{prop:mms-no-ef1}
\\
\hline
\textbf{Multiplicative MMS} 
&  None \eqref{prop:cardinal-no-ordinal}
&  -
& None \eqref{prop:mms-no-ef1}
& None \eqref{prop:mms-no-ef1}
\\
\hline
\textbf{EF1} 
&  None \eqref{prop:ef1-no-ordinal}
&  [not checked]
&  -
&  - \\
\hline
\textbf{EFx} 
&  At most $n/2$ \eqref{prop:efx-implies-ordinal}
&  [not checked]
&  -
&  -
\\
\hline
\end{tabular}
\caption{
\label{tab:relations}
High-level summary of implication relations between fairness notions for chores.
Each cell indicates to what extent the fairness notion in the row implies the fairness notion in the column.
The number in parentheses is the proposition number.
}
\end{table*}

\begin{proposition}
\label{prop:cardinal-no-ordinal}
 For any integers $n, k, d\geq 2$, the ordinal approximation 1-out-of-$d$ MMS is not implied by cardinal approximation $(1+1/k)$-fraction 1-out-of-$n$ MMS.
\end{proposition}
\begin{proof}
Given integers $n,k\geq 2$, consider an instance with $n$ agents with identical valuations, 
one ``hard'' chore with value $-k$, 
and $k$ ``easy'' chores with value $-1$.
For any integers $d,n\geq 2$, we have $MMS^d = MMS^n = -k$, since one bundle must hold the hard chore.

Consider an allocation in which agent 1 gets the hard chore and one easy chore, and the other $k-1$ easy chores are divided among agents $2,\ldots,n$ in a balanced way (e.g. using round-robin).
The value of agent 1 is $-k-1 = (1+1/k)\cdot MMS^n$.
The value of every other agent is at least $(k-1)\cdot (-1) = -k+1 >MMS^n$. 
So the allocation satisfies the cardinal approximation $(1+1/k)$-fraction 1-out-of-$n$ MMS. 
However, the allocation does not satisfy 1-out-of-$d$ MMS for agent 1.
\end{proof}

\begin{proposition}
\label{prop:ef1-no-ordinal}
For any $n, d\geq 2$, the ordinal approximation 1-out-of-$d$ MMS is not implied by EF1 or PROP1 allocation among $n$ agents.
\end{proposition}
\begin{proof}
Consider the instance and allocation of Prop. \ref{prop:cardinal-no-ordinal}, with $k=n$. The allocation is EF1, since every agent $2,\ldots,n$ receives one chore and is not envious. Agent 1 too does not envy after removing the hard chore. Similarly, the allocation satisfies proportionality after removing the hard chore. As mentioned above, it does not satisfy 1-out-of-$d$ MMS for agent 1.
\end{proof}

\begin{proposition}
\label{prop:mms-no-ef1}
For any integers $n\geq 3$ and $c\geq 1$, a 1-out-of-$n$ MMS allocation (even without approximations) does not guarantee any positive approximation of proportionality-up-to-$c$ items or envy-freeness-up-to-$c$ items.
\end{proposition}
\begin{proof}
Consider the instance of Prop. \ref{prop:cardinal-no-ordinal} with $k = 3c+3$.
Consider an allocation in which agent 1 gets the hard chore, agent 2 gets all $k$ easy chores, and agents $3,\ldots,n$ get nothing.
The value of every agent is at least $-k = MMS^n$, so the allocation satisfies 1-out-of-$n$ MMS fairness.
However, agent 2 envies agent 3 even after removing $k-1>c$ chores.
Moreover, the proportional value of all agents is $-2k/n \geq -\frac{2}{3}k = -(2c+2)$, 
so the allocation is not proportional for agent 2 even after removing $c = k/3-1$ chores.
\end{proof}

\begin{proposition}
\label{prop:ordinal-implies-cardinal}
For every integers $d\geq 1$ and $q\geq 1$,
if an allocation satisfies $1$-out-of-$d$ MMS, then it gives each agent at least $q$ times his $1$-out-of-$q d$ MMS.
The factors are tight.
\end{proposition}
\begin{proof}
Let $X_1,\ldots,X_{qd}$ be any partition of a set of chores into $q d$ parts, and let $v_{qd} := \min_j v(X_j)$.
Grouping the parts arbitrarily into $d$ groups of $q$ subsets yields a $d$-partition $Y_1,\ldots,Y_{d}$ in which the value of each part is at least $q\cdot v_{qd}$.
Therefore, the $MMS^d$ of any agent is at least $q$ times the agent's $MMS^{qd}$.

For tightness, consider an instance with $q d+1$ chores of value $-1$. Then $MMS^d = -q-1$ and $MMS^{qd+1}=-1$, so $MMS^d < q\cdot MMS^{qd+1}$.
\end{proof}
The most useful implication of the above proposition is for $q=2$, where it implies that the ordinal 1-out-of-$n/2$ MMS implies a $2$-factor multiplicative approximation of the 1-out-of-$n$ MMS.

\begin{proposition}
\label{prop:efx-implies-ordinal}
For any $n\geq 2$, an EFx allocation among $n$ agents (when it exists) satisfies 1-out-of-$(n+1)/2$ MMS, 
and may not satisfy 1-out-of-$d$ MMS for $d>(n+1)/2$.
\end{proposition}
\begin{proof}
Let $A_1,\ldots,A_n$ be an EFx allocation. 
We focus on agent 1 and prove that $v_1(A_1)\geq MMS^d_1$ for any $d\leq (n+1)/2$ (the proof for other agents is identical).

Let $c_1$ be a chore that maximizes $v_1$ in $A_1$. 
If $A_1$ contains only $c_1$, then we are done, since any $d$-partition must have at least one bundle that contains $c_1$, so its value must be at most $v_1(A_1)$.

Otherwise, let $A_1 = \{c_1\}\cup A'_1$. 
EFx means that $v_1(A'_1)\geq v_1(A_j)$  for all $j\in[n]$.
Since $c_1$ is highest-valued in $A_1$, we have $v_1(c_1)\geq v_1(A'_1)$.
So, in the $(n+1)$-partition of $M$ into $c_1, A_1', A_2,\ldots A_n$, the first two elements $c_1, A_1'$ have the highest value. This means that $v_1(A_1)\geq \frac{2}{n+1} v_1(M) \geq MMS_1^{\frac{n+1}{2}}$.

For tightness, consider an instance with two hard chores of value $-1$, and many easy chores of total value $-n+1$ (for the sake of the proof, we can consider these chores as divisible). The total value of all chores is $-n-1$.
For any $d>(n+1)/2$, we can partition the chores into $d$ bundles of equal value, which is $(-n-1)/d > -2$.
The allocation in which agent 1 gets the two hard chores and each of the other agents gets easy chores of total value $-1$ is EFx, but does not satisfy 1-out-of-$d$ MMS for agent 1.
\end{proof}

\section{Ordinal approximation for mixed items} \label{app:mixed}
When the items can be a mixture of goods and chores, no multiplicative approximation of the MMS is guaranteed to exist \citep{kulkarni2021approximating}. This raises the question of whether an ordinal approximation of the MMS can be guaranteed.
The answer is no.
\begin{proposition}
There is an instance with $n=3$ agents with mixed valuations, in which no allocation is 1-out-of-$d$ MMS, for any positive integer $d$.
\end{proposition}
\begin{proof}
\citet{kulkarni2021approximating}
show an instance with three agents and 15 items,
where 12 items are goods for everyone and 3 items are chores for everyone. This instance has the following properties:
\begin{enumerate}
\item For each agent, the sum of values of all 15 items is positive.
\item In any allocation, at least one agent gets a negative value. 
\end{enumerate}
Property 1 implies that the 1-out-of-$d$ MMS is at least $0$ for any $d\geq 1$, since one can put all items in one bundle and have $d-1$ empty bundles. 
Property 2 then implies that no allocation satisfies $1$-out-of-$d$ MMS thresholds.
\end{proof}

\section{The Size of an Instance}\label{app:size}
The next lemma states that, without loss of generality, we need only focus on instances with $m \geq 2 d$. Otherwise, an ordered chores instance can be reduced by giving $2d - m$ agents exactly one chore from $\{c_{1}, \ldots, c_{2d-m}\}$, which will be worth at least $\text{MMS}_{i}^{d}(M)$.\footnote{A special case of this lemma with $d=n$ for goods was previously shown in \citep{amanatidis2017approximation}.}


\begin{lemma}\label{lem:mlessthan2n}
Given an ordered chores instance $I = \ins{N, M, V}$ such that $m < 2d$, for each agent $i\in N$,
\begin{align*}
\text{MMS}_{i}^{d-1}(M \setminus \{c_{1}\}) \geq \text{MMS}_{i}^{d}(M).
\end{align*} 
\end{lemma}

\begin{proof}
Let $I = \ins{N, M, V}$ be an ordered instance. Then, $c_1$ is the largest (least preferred) chore  for all agents. We show that when $m < 2d$, chore $c_{1}$ can always be in a bundle by itself, and thus, taking the remaining $d-1$ bundles forms a $\text{MMS}_{i}^{d-1}(M \setminus \{c_{1}\})$ partition which is worth at least $\text{MMS}_{i}^{d}(M)$.

Consider any $\text{MMS}_{i}^{d}(M)$ partition $(A_{1}, \ldots, A_{d})$. If $c_{1}$ is already in a bundle by itself, then by definition we have that $\text{MMS}_{i}^{d-1}(M \setminus \{c_{1}\}) \geq \text{MMS}_{i}^{d}(M)$.
Consider the case that $c_{1}$ is in bundle $A_{j}$. Since $m < 2d$, there must be some bundle of any $\text{MMS}_{i}^{d}(M)$ partition which contains only one chore. Let $c_k$ denote one such lonely chore.
Because the instance is ordered $v_{i}(c_{1}) \leq v_{i}(c_{k})$. 
Since each additional chore in $A_{j}$ has non-positive value,  $v_{i}(A_{j}) \leq v_{i}(c_{1}) \leq v_{i}(c_{k})$. 
Swap $c_{1}$ and $c_{k}$ to form bundles $A_{j}' = (A_{j} \setminus \{c_{1}\}) \cup \{c_{k}\}$ and $\{c_{1}\}$. 
Since $v_{i}(c_{1}) \leq v_{i}(c_{k})$ we have  $v_{i}(A_{j}') = v_{i}(A_{j}) - v_{i}(c_{1}) + v_{i}(c_{k}) \geq v_{i}(A_{j})$. Likewise $v_{i}(c_{1}) \geq v_{i}(A_{j})$. Thus after swapping $c_{1}$ and $c_{k}$, the value of the minimum bundle does not decrease and $c_{1}$ is in a bundle by itself, implying $\text{MMS}_{i}^{d-1}(M \setminus \{c_{1}\}) \geq \text{MMS}_{i}^{d}(M)$.
\end{proof}

\end{document}